\documentclass[12pt,draftcls,journal,onecolumn]{IEEEtran}

\usepackage{amssymb,amsthm, amsmath,latexsym}
\usepackage{graphicx}
\usepackage{mathrsfs}
\usepackage{amsfonts}
\usepackage{amssymb}
\usepackage{longtable}
\usepackage{amsmath}
\usepackage{setspace}
\usepackage{caption}
\usepackage[figuresright]{rotating}
\IfFileExists{ifsym.sty}{\usepackage[misc]{ifsym}}{}
\IfFileExists{bbm.sty}{\usepackage{bbm}}{}
\usepackage{makecell}
\usepackage{arydshln}
\usepackage{supertabular}
\usepackage{booktabs}
\usepackage{color}

\newtheorem{theorem}{Theorem}
\newtheorem{lemma}[theorem]{Lemma}
\newtheorem{remark}[theorem]{Remark}

\newtheorem{corollary}[theorem]{Corollary}

\newtheorem{definition}[theorem]{Definition}

\newcommand{\F}{{\mathbb{F}}}

\usepackage{blindtext}

\ifCLASSINFOpdf

\else

\fi

\begin{document}
%
\title{ Constructions of LCPs and LCD codes from twisted Reed-Solomon codes
}
\author{Shuo Sun, Wenwen Chen, Chao Liu, Yaozong Zhang$^{*}$, Xiaoqiang Wang
}

\renewcommand{\thefootnote}{\empty}
\footnotetext{\thanks{~*Corresponding author. }
\newline \indent ~Shuo Sun is with the School of Mathematics and Statistics, Central China Normal University (E-mail: sunnshuoo@163.com).
\newline \indent ~Wenwen Chen, Chao Liu, Yaozong Zhang, Xiaoqiang Wang are with the Hubei Key Laboratory of Applied Mathematics, Faculty of Mathematics and Statistics, Hubei University, Wuhan 430062, China (E-mail: chewewe@163.com; chliuu@163.com; zyzsdutms@163.com; waxiqq@163.com).
}

\maketitle

\begin{abstract}
Linear complementary pairs (LCPs) and linear complementary dual (LCD) codes have important applications in orthogonal direct-sum masking (ODSM), which provides effective countermeasures against side-channel attacks and fault-injection attacks. While LCD codes have been extensively investigated, comparatively fewer results are available for general LCPs. In this paper, we further investigate LCPs of twisted Reed--Solomon (TRS) codes. We derive necessary conditions for two TRS codes to form an LCP and establish several sufficient conditions and explicit constructions. We also study LCD codes constructed from TRS codes and investigate the security parameters of the resulting LCPs. Furthermore, under suitable conditions, we obtain MDS LCPs of TRS codes.
\end{abstract}

%
\textbf{2020 Mathematics Subject Classification:} 94B05, 94B15.

\textbf{Keywords:} LCPs, LCD codes, MDS codes, twisted Reed--Solomon codes.

%
\IEEEpeerreviewmaketitle

\begingroup\small
\noindent\textbf{Review key.} Red text and tags E01--E03 identify confirmed
source or bibliographic errors. E01 marks an invisible duplicate label at two
section headings; the heading text itself is correct. Existing blue text is
from the original manuscript. Suggested mathematical clarifications N01--N02
are recorded in source comments and in the separate revision report.
\par\endgroup\medskip

\section{Introduction}\label{sec-introduction}

\renewcommand{\thetheorem}{\arabic{section}.\arabic{theorem}}
\setcounter{theorem}{0}

Let $\F_q$ be the finite field with $q$ elements, where $q$ is a prime power. An $[n,k,d]_q$ linear code $C$ is a $k$-dimensional subspace of $\F_q^n$ with minimum Hamming distance $d$.
Let \(C,D\subseteq\F_q^n\) be linear codes of dimensions \(k\) and
\(n-k\), respectively. The dual code of a linear code \(C\), denoted by
\(C^\perp\), is defined as follows:
\[
C^\perp=\{\mathbf b\in\F_q^n:\mathbf b\cdot\mathbf c=0,\
\text{for all }\mathbf c\in C\},
\]
where ``\(\cdot\)'' denotes the Euclidean inner product. If \(C\) and
\(D\) are linear codes over \(\F_q\) with the same length \(n\) and
respective dimensions \(k\) and \(n-k\) such that
\[
C\oplus D=\F_q^n,
\]
then the pair \((C,D)\) is called a linear complementary pair (LCP) of
codes.
When $D=C^\perp$, the LCP condition becomes $C\cap C^\perp=\{\mathbf{0}\}$, and $C$ is called a linear complementary dual (LCD) code.


LCP and LCD codes arise naturally in the implementation of cryptographic countermeasures. Even when a cryptographic algorithm is secure from a theoretical point of view, its physical implementation may leak information through side channels or may be vulnerable to fault-injection attacks. The orthogonal direct-sum masking scheme introduced in \cite{Bringer2014,CarletGuilley2016} uses two complementary subspaces $C$ and $D$ to encode sensitive data and masks. An important security parameter of such a construction is
\[
\min\{d(C),d(D^\perp)\}.
\]
In the LCD case $D=C^\perp$, this quantity reduces to $d(C)$.

LCD codes were introduced by Massey \cite{Massey1992} and have subsequently been investigated in many different settings; see, for example, \cite{CarletEtAl2018LCD,EsmaeiliYari2009,GalindoEtAl2019,WuLee2020,ZhouEtAl2019}. By comparison, the theory of general LCPs is less developed. Carlet et al. \cite{CarletEtAl2018LCP} studied LCPs arising from constacyclic and quasi-cyclic codes. Lobillo and Mu\~noz \cite{LobilloMunoz2023} considered skew constacyclic codes, while Bhowmick et al. \cite{Bhowmick2023} investigated LCPs constructed from algebraic-geometry codes. From the viewpoint of security parameters, Carlet et al. \cite{CarletEtAl2019Sigma} showed that, for $q>2$, LCPs can attain the best possible minimum-distance behavior for suitable parameters. Optimal binary LCPs constructed from Solomon--Stiffler codes were obtained in \cite{Guneri2023}.

Another active direction in coding theory concerns twisted Reed--Solomon codes. Classical generalized Reed--Solomon codes form one of the most important families of MDS codes. Twisted Reed--Solomon codes, introduced in \cite{BeelenPuchingerRosenkilde2022}, modify the usual polynomial evaluation space by adding carefully chosen high-degree terms whose coefficients depend on selected low-degree coefficients. These constructions can produce MDS codes that are not equivalent to generalized Reed--Solomon codes. Their structural and cryptographic properties have therefore received considerable attention \cite{BeelenEtAl2018,HuangYueNiu2023,LiuLiu2021,LavauzelleRenner2020,ZhangZhouTang2022,ZhuLiao2021}.

Most previous work on LCD twisted Reed--Solomon codes focuses either on a single twist or on particular configurations of multiple twists. In the present paper, we consider a more general multiple-twist setting and investigate when two twisted Reed--Solomon codes form an LCP.

Our approach is based on a simple observation. If
\(
\mathbf{V}_j=(\alpha_1^j,\alpha_2^j,\ldots,\alpha_n^j),
\)
with \(\alpha_1,\ldots,\alpha_n\) are distinct evaluation points, then
\(
\mathbf{V}_0,\mathbf{V}_1,\ldots,\mathbf{V}_{n-1}
\)
form a basis of \(\mathbb F_q^n\), since the corresponding Vandermonde
matrix is nonsingular. Therefore, every generator row of a TRS code can
be uniquely expressed as a linear combination of these basis vectors.
Consequently, the stacked generator matrix can be written as
\(
G=AV,
\)
where \(A\) is the coefficient matrix with respect to the Vandermonde
basis \( \{\mathbf V_0,\ldots,\mathbf V_{n-1}\}\) and \(V\) is the
nonsingular Vandermonde matrix. Hence,
\(
G\text{ is nonsingular if and only if }A\text{ is nonsingular}.
\)
Therefore, the LCP problem can be reduced to studying the coefficient
matrix \(A\). This viewpoint avoids repeated Vandermonde computations
and makes the role of the exponents of the twist terms in the generator
matrices transparent.

The remainder of the paper is organized as follows. Section~II recalls twisted Reed--Solomon codes and a basic matrix
criterion for LCPs. Section~III derives necessary exponent conditions and establishes several sufficient conditions for constructing LCPs of twisted Reed--Solomon codes. MDS LCPs are then obtained by combining these results with a known subfield-chain construction. Section~IV specializes the method to dual pairs and gives several constructions of LCD twisted Reed--Solomon codes. Finally, Section~V concludes the paper.

\section{Preliminaries}\label{sec-preliminaries}

\renewcommand{\thetheorem}{\arabic{section}.\arabic{theorem}}
\setcounter{theorem}{0}

Let $\F_q$ be the finite field of $q$ elements, where $q$ is an odd prime power and $\F_q^*=\F_q\setminus\{0\}$. Let $\F_q[x]$ be the polynomial ring over $\F_q$.
Let
$
\boldsymbol{\alpha}=(\alpha_1,\alpha_2,\ldots,\alpha_n)
$
be a vector of pairwise distinct elements of $\F_q$, and let
\[
\boldsymbol{v}
=
(v_1,\ldots,v_n)\in(\F_q^*)^n.
\]
Define the evaluation map
\[
\operatorname{ev}_{\boldsymbol{\alpha},\boldsymbol{v}}
:\F_q[x]\longrightarrow\F_q^n
\]
by
\[
f(x)\longmapsto\bigl(v_1f(\alpha_1),v_2f(\alpha_2),\ldots,v_nf(\alpha_n)\bigr).
\]
\begin{definition}
Let \(k\) be a positive integer and let \(0\le \ell\le \min\{k,n-k\}. \)
Let \(\mathbf{h} = (h_1, h_2, \ldots, h_\ell)\) \text{and} \(\mathbf{t} = (t_1, t_2, \ldots, t_\ell)\)
be two vectors of length \(\ell\) over \(\mathbb{Z}\), where
\(h_1,\ldots,h_\ell\) are pairwise distinct,
\(t_1,\ldots,t_\ell\) are pairwise distinct, and
\(0\le h_i\le k-1,\; 1\le t_i\le n-k,\; 1\le i\le \ell\).
Let
\(\boldsymbol{\eta}
=(\eta_1,\eta_2,\ldots,\eta_\ell)
\in(\F_q^*)^\ell.\)
When \(\ell=0\), the vectors
\(\mathbf{h}\), \(\boldsymbol{t}\), and \(\boldsymbol{\eta}\)
are understood to be empty, and the corresponding sum of twist terms
is understood to be zero.

The set of
\((\boldsymbol{t},\mathbf{h},\boldsymbol{\eta})\)-twisted
polynomials is defined as
\[
\mathcal{P}_{\boldsymbol{t},\mathbf{h},\boldsymbol{\eta}}^{n,k}
=
\left\{
\sum_{i=0}^{k-1}f_i x^i
+
\sum_{j=1}^{\ell}
\eta_j f_{h_j}x^{k-1+t_j}
:
f_i\in\F_q
\right\},
\]
which is a \(k\)-dimensional \(\F_q\)-linear subspace of
\(\F_q[x]\).

For each \(1\le j\le\ell\), the term
\[
\eta_j f_{h_j}x^{k-1+t_j}
\]
is called a twist term, where \(h_j\) is the corresponding hook,
\(t_j\) is the twist parameter, and \(\eta_j\) is the twist
coefficient. The exponent \(k-1+t_j\) is called the corresponding
twist exponent.
\end{definition}


\begin{definition}
Let the notation be as in Definition~2.1. The twisted generalized
Reed--Solomon (TGRS) code of length \(n\) and dimension \(k\) is defined by
\[
C_{n,k}(\boldsymbol{\alpha},\mathbf{v},
\boldsymbol{t},\mathbf{h},\boldsymbol{\eta})
=
\left\{
\bigl(
v_1f(\alpha_1),v_2f(\alpha_2),\ldots,v_nf(\alpha_n)
\bigr)
:
f\in
\mathcal{P}_{\boldsymbol{t},\mathbf{h},\boldsymbol{\eta}}^{n,k}
\right\}.
\]
When
\(
\mathbf{v}=\mathbf{1}=(1,\ldots,1),
\)
the code is called a twisted Reed--Solomon (TRS) code.

When \(\ell=0\), the TGRS code reduces to a generalized
Reed--Solomon (GRS) code; in particular, when
\(\mathbf{v}=\mathbf{1}\), it reduces to a Reed--Solomon (RS) code.
\end{definition}

We record the generator matrix explicitly, since its row structure is the key to the arguments in the next section. For $0\le i\le k-1$, define the row vector
\[
\boldsymbol{g}_i=
\begin{cases}
(v_1\alpha_1^i,\ldots,v_n\alpha_n^i),
& i\notin\{h_1,\ldots,h_\ell\},\\[1mm]
\bigl(v_1(\alpha_1^{h_j}+\eta_j\alpha_1^{k-1+t_j}),\ldots,
      v_n(\alpha_n^{h_j}+\eta_j\alpha_n^{k-1+t_j})\bigr),
& i=h_j.
\end{cases}
\]
Then a generator matrix is
\begin{equation}\label{eqggg}
G_{n,k}(\boldsymbol{\alpha},\boldsymbol{v},
\boldsymbol{t},\boldsymbol{h},\boldsymbol{\eta})
=
\begin{pmatrix}
\boldsymbol{g}_0\\
\boldsymbol{g}_1\\
\vdots\\
\boldsymbol{g}_{k-1}
\end{pmatrix}.
\end{equation}
Thus a twist does not add a new generator row; rather, it replaces the ordinary evaluation row of degree $h_j$ by that row plus a multiple of the evaluation row of degree $k-1+t_j$.

For later use, define
\[
\mathbf{V}_m
=
(\alpha_1^m,\alpha_2^m,\ldots,\alpha_n^m),
\]
\(
\text{where } 0\le m\le n-1.
\)

For a TRS code \((\boldsymbol{v}=\mathbf{1})\), the generator rows in (\ref{eqggg}) become
\begin{equation}\label{eqggg2}
\boldsymbol{g}_i=
\begin{cases}
\mathbf{V}_i, & i\notin\{h_1,\ldots,h_\ell\},\\
\mathbf{V}_{h_j}+\eta_j\mathbf{V}_{k-1+t_j}, & i=h_j.
\end{cases}
\end{equation}
The compact representation (\ref{eqggg2}) is the form used in the sequel. Since $\alpha_1,\ldots,\alpha_n$ are pairwise distinct, the Vandermonde matrix
\[
V=
\begin{pmatrix}
\mathbf{V}_0\\
\mathbf{V}_1\\
\vdots\\
\mathbf{V}_{n-1}
\end{pmatrix}
\]
is nonsingular. Consequently,
\(
\{\mathbf{V}_0,\mathbf{V}_1,\ldots,\mathbf{V}_{n-1}\}
\)
is a basis of $\F_q^n$.

\begin{definition}
Let $C,D\subseteq\F_q^n$ be linear codes. The pair $(C,D)$ is an LCP if
\(
C\cap D=\{\mathbf{0}\} \text{ and } C+D=\F_q^n.
\)
If $D=C^\perp$, then $C$ is an LCD code.
\end{definition}

\begin{lemma}\label{eq:10821}
Let $C$ and $D$ be linear codes of length $n$ and dimensions $k$ and $n-k$, respectively. Let $G_C$ and $G_D$ be generator matrices of $C$ and $D$. Then $(C,D)$ is an LCP if and only if
\[
G=
\begin{pmatrix}
G_C\\
G_D
\end{pmatrix}
\]
is nonsingular.
\end{lemma}

\begin{proof}
The rows of $G_C$ form a basis of $C$, while the rows of $G_D$ form a basis of $D$. Hence $G$ is nonsingular if and only if these $n$ rows are linearly independent. Since $\dim C+\dim D=n$, this is equivalent to $C\cap D=\{\mathbf{0}\}$, and therefore to $C+D=\F_q^n$.
\end{proof}

\section{LCPs of twisted Reed-Solomon codes}

\renewcommand{\thetheorem}{\arabic{section}.\arabic{theorem}}
\setcounter{theorem}{0}

Throughout this section, assume
\(
k\le n-k.
\)
Consider the two TRS codes
\begin{equation}\label{eq:CD}
C=C_{n,k}\bigl(\boldsymbol{\alpha},\mathbf{1},
\boldsymbol{t},\boldsymbol{h},\boldsymbol{\eta}\bigr)
\text{ and }
D=C_{n,n-k}\bigl(\boldsymbol{\alpha},\mathbf{1},
\boldsymbol{\gamma},\boldsymbol{b},\boldsymbol{\delta}\bigr),
\end{equation}
where
\(
\boldsymbol{h}=(0,1,\ldots,r-1)
\text{ and }
\boldsymbol{b}=(0,1,\ldots,s-1),
\)
with
\(0\le r\le k\), \(0\le s\le k\).
Here, when \(r=0\) or \(s=0\), the corresponding vector
\(\boldsymbol{h}\) or \(\boldsymbol{b}\), respectively, is understood
to be empty.

By~(\ref{eqggg2}), the generator matrices of \(C\) and \(D\) take
the forms
\begin{equation*}
G_C=
\begin{pmatrix}
\mathbf{V}_0+\eta_1\mathbf{V}_{k-1+t_1}\\
\mathbf{V}_1+\eta_2\mathbf{V}_{k-1+t_2}\\
\vdots\\
\mathbf{V}_{r-1}+\eta_r\mathbf{V}_{k-1+t_r}\\
\mathbf{V}_r\\
\vdots\\
\mathbf{V}_{k-1}
\end{pmatrix},
\end{equation*}
and
\begin{equation*}
G_D=
\begin{pmatrix}
\mathbf{V}_0+\delta_1\mathbf{V}_{n-k-1+\gamma_1}\\
\mathbf{V}_1+\delta_2\mathbf{V}_{n-k-1+\gamma_2}\\
\vdots\\
\mathbf{V}_{s-1}+\delta_s\mathbf{V}_{n-k-1+\gamma_s}\\
\mathbf{V}_s\\
\vdots\\
\mathbf{V}_{n-k-1}
\end{pmatrix}.
\end{equation*}

Equivalently, their rows are
\begin{equation}\label{eq:0821}
\mathbf{C}_i=
\begin{cases}
\mathbf{V}_{i-1}+\eta_i\mathbf{V}_{k-1+t_i},&1\le i\le r,\\
\mathbf{V}_{i-1},&r<i\le k,
\end{cases}
\end{equation}
and
\begin{equation}\label{eq:082101}
\mathbf{D}_j=
\begin{cases}
\mathbf{V}_{j-1}+\delta_j\mathbf{V}_{n-k-1+\gamma_j},&1\le j\le s,\\
\mathbf{V}_{j-1},&s<j\le n-k.
\end{cases}
\end{equation}

These matrices make the structure of the later proofs visible: the ordinary parts occupy the Vandermonde positions \(0,1,\ldots\), while the twist terms create additional entries at the positions \(k-1+t_i\) and \(n-k-1+\gamma_j\).

\begin{lemma}
Assume \(0\le r<k\). If \((C,D)\) is an LCP, then \(s=k\). Consequently,
\[
\{\gamma_1,\ldots,\gamma_k\}=\{1,\ldots,k\}.
\]
\end{lemma}

\begin{proof}
Suppose \(s<k\). Since \(r<k\), the vector \(\mathbf{V}_{k-1}\) occurs as an untwisted generator row of \(C\). Because \(k\le n-k\) and \(s<k\), the same vector \(\mathbf{V}_{k-1}\) also occurs as an untwisted generator row of \(D\). Hence,
\[
\mathbf{V}_{k-1}\in C\cap D,
\]
contradicting \(C\cap D=\{\mathbf{0}\}\). Therefore, \(s\ge k\).

By Definition~2.1, for a TRS code of dimension \(n-k\),
the twist parameters satisfy
\[
1\le \gamma_i\le n-(n-k)=k.
\]
Since the \(\gamma_i\) are pairwise distinct, \(s\le k\). Hence \(s=k\). Since \(\gamma_1,\ldots,\gamma_k\) are \(k\) distinct elements of \(\{1,\ldots,k\}\), their set is exactly \(\{1,\ldots,k\}\).
\end{proof}

\begin{remark}
The equality
\[
\{\gamma_1,\ldots,\gamma_k\}=\{1,\ldots,k\}
\]
does not by itself imply \(\gamma_i=i\) for each \(i\).
The parameter \(\gamma_i\) is paired with the hook \(b_i=i-1\)
and the coefficient \(\delta_i\). Therefore, whenever an argument
requires the \(i\)-th hook to be paired specifically with the exponent
\(n-k-1+i\), the condition \(\gamma_i=i\) will be imposed explicitly.
\end{remark}

\begin{lemma}
If \((C,D)\) is an LCP, then the union of the supports of the
generator rows of \(C\) and \(D\), with respect to the Vandermonde
basis \(\{\mathbf{V}_0,\ldots,\mathbf{V}_{n-1}\}\), is
\(
\{0,1,\ldots,n-1\}.
\)
In particular,
\[
\{n-k,n-k+1,\ldots,n-1\}\subseteq E,
\]
where
\(
E=
\{k-1+t_i:1\le i\le r\}
\cup
\{n-k-1+\gamma_j:1\le j\le s\}.
\)
\end{lemma}

\begin{proof}
Every row of \(G_C\) and \(G_D\), when expressed with respect to the
Vandermonde basis
\(
\{\mathbf{V}_0,\mathbf{V}_1,\ldots,\mathbf{V}_{n-1}\},
\)
has its support contained in
\[
S=\{0,\ldots,k-1\}
\cup\{0,\ldots,n-k-1\}
\cup E.
\]
Since \(1\le t_i\le n-k\) and \(1\le\gamma_j\le k\), we have
\[
S\subseteq\{0,1,\ldots,n-1\}.
\]

Suppose that some \(e\in\{0,1,\ldots,n-1\}\) does not belong to \(S\).
Then every row of the stacked generator matrix
\[
\begin{pmatrix}
G_C\\
G_D
\end{pmatrix}
\]
lies in the subspace
\(
\operatorname{span}\{\mathbf{V}_e:e\in S\}.
\)

Since \(|S|<n\) and
\(\mathbf{V}_0,\ldots,\mathbf{V}_{n-1}\) form a basis of
\(\mathbb{F}_q^n\), this subspace has dimension
\[
\dim\operatorname{span}\{\mathbf{V}_e:e\in S\}
=|S|<n.
\]
Hence the stacked generator matrix is singular, contradicting
Lemma~\ref{eq:10821}. Therefore,
\[
S=\{0,1,\ldots,n-1\}.
\]

The ordinary Vandermonde components of the generator rows already
account for all indices in
\(
\{0,1,\ldots,n-k-1\}.
\)
Consequently, every remaining index in
\(\{n-k,\ldots,n-1\}\) must belong to \(E\). Hence
\(
\{n-k,n-k+1,\ldots,n-1\}\subseteq E.
\)
\end{proof}

The preceding proof shows that the exponents of the twist terms appearing in the generator matrices play a decisive role in the LCP property.

In the following conditions, the equality $\gamma_i=i$
is imposed as an ordering condition rather than a consequence of
$\{\gamma_1,\ldots,\gamma_k\}=\{1,\ldots,k\}$.

\begin{theorem}\label{thm:01}
Let \(C\) and \(D\) be as above, and suppose
\[
\{n-k,\ldots,n-1\}
\subseteq
\{k-1+t_i:1\le i\le r\}
\cup
\{n-k-1+\gamma_j:1\le j\le s\}.
\]
Then \((C,D)\) is an LCP whenever one of the following conditions holds:

\begin{enumerate}
\renewcommand{\labelenumi}{(\roman{enumi})}

\item
\(0\le r<k\), \(s=k\), and
\(1\le t_i\le n-2k\) for \(1\le i\le r\).

\item
\(r=k\), all the exponents of the twist terms appearing in the
generator matrices of \(C\) and \(D\), namely
\[
k-1+t_1,\ldots,k-1+t_k, n-k-1+\gamma_1,\ldots,n-k-1+\gamma_s,
\]
are pairwise distinct, and
\(k-1+t_i\le n-k-1\) for \(1\le i\le s\), while
\(k-1+t_i\ge n-k\) for \(s+1\le i\le k\).

\item
\(0\le r<k\), \(s=k\),
\(1\le t_1<\cdots<t_r\), and
\(\gamma_i=i\) for \(1\le i\le k\).
Moreover, there exists \(j\) with \(1\le j\le r-1\) such that
\(
t_j\le n-2k<t_{j+1},
\)
and
$t_i=n-2k+i\text{ and }\eta_i\ne\delta_i$
for every \(j+1\le i\le r\).
\item
\(0\le r<k\), \(s=k\),
\(\gamma_i=i\) for \(1\le i\le k\), and $t_i=n-2k+i\text{ and }
\eta_i\ne\delta_i$ for every \(1\le i\le r\).

\item
\(r=k\), \(1\le t_1<\cdots<t_k\),
\(0\le s\le k\), \(n\le 2k-1+t_1\), and $\gamma_i=i\text{ and }
\eta_i\ne\delta_i$ for every \(1\le i\le s\).
\end{enumerate}
\end{theorem}

\begin{proof}
From Lemma~\ref{eq:10821}, it suffices to prove that
\[
G=
\begin{pmatrix}
G_C\\
G_D
\end{pmatrix}
\]
is nonsingular. Since
\[
V=
\begin{pmatrix}
\mathbf{V}_0\\
\vdots\\
\mathbf{V}_{n-1}
\end{pmatrix}
\]
is a nonsingular Vandermonde matrix, one may write \(G=AV\) for a uniquely determined coefficient matrix \(A\). Hence \(G\) is nonsingular if and only if \(A\) is nonsingular.

Equivalently, consider a linear relation
\begin{equation}\label{eq:ddw8}
\sum_{i=1}^{k}a_i\mathbf{C}_i
+
\sum_{j=1}^{n-k}b_j\mathbf{D}_j
=0,
\end{equation}
where \(\mathbf{C}_i\) and \(\mathbf{D}_j\) are given in
(\ref{eq:0821}) and (\ref{eq:082101}), respectively.
Because \(\mathbf{V}_0,\ldots,\mathbf{V}_{n-1}\) form a basis of \(\F_q^n\), the coefficient of each Vandermonde basis vector must vanish separately.

\medskip
\noindent\textbf{Case (i).}
Since \(s=k\) and \(\gamma_1,\ldots,\gamma_k\) are \(k\) pairwise distinct elements of \(\{1,\ldots,k\}\), we have
\(
\{\gamma_1,\ldots,\gamma_k\}=\{1,\ldots,k\}.
\)
Then the exponents of the twist terms in the generator matrix of \(D\) are exactly
\(
\{n-k,n-k+1,\ldots,n-1\},
\)
although not necessarily in this order.

On the other hand, by the assumption \(1\le t_i\le n-2k\),
\(
k-1+t_i\le n-k-1,
\)
so none of the exponents of the twist terms appearing in the generator matrix of \(C\) belongs to the set \(\{n-k,n-k+1,\ldots,n-1\}\). Comparing the coefficients corresponding to the exponents \(n-k,n-k+1,\ldots,n-1\) gives
$
b_j\delta_j=0,
$
where \(1\le j\le k\).
Since every \(\delta_j\) is nonzero, we have
\(
b_1=\cdots=b_k=0.
\)

We then consider the coefficients of the Vandermonde basis vectors
\(
\mathbf{V}_0,\ldots,\mathbf{V}_{k-1}.
\)
By the assumption \(1\le t_i\le n-2k\), all twist exponents appearing in \(G_C\) satisfy
\(
k\le k-1+t_i\le n-k-1,
\)
whereas all twist exponents appearing in \(G_D\) satisfy
\(
n-k\le n-k-1+\gamma_i\le n-1.
\)
Therefore, none of the twist terms in \(G_C\) or \(G_D\) contributes to the coefficients of
\(
\mathbf{V}_0,\ldots,\mathbf{V}_{k-1}.
\)
Hence, for each \(1\le i\le k\), the coefficient of \(\mathbf{V}_{i-1}\) in (\ref{eq:ddw8}) is
\(
a_i+b_i.
\)
Since \(b_1=\cdots=b_k=0\), we obtain
\(
a_1=\cdots=a_k=0.
\)

Finally, the remaining rows of \(G_D\) that do not contain twist terms correspond to the remaining Vandermonde basis vectors. Therefore,
\(
b_{k+1}=\cdots=b_{n-k}=0.
\)
Thus, (\ref{eq:ddw8}) has only the zero solution.

\medskip
\noindent\textbf{Case (ii).}
By the assumed exponent ranges, the first \(s\) twist terms in the
generator matrix of \(C\) have exponents in
\(
\{k,\ldots,n-k-1\},
\)
whereas the remaining \(k-s\) twist terms in \(C\) have exponents in
\(
\{n-k,\ldots,n-1\}.
\)

On the other hand, since \(1\le \gamma_j\le k\), every twist exponent
appearing in \(D\) satisfies
\(
n-k\le n-k-1+\gamma_j\le n-1,
\)
where \(1\le j\le s\).
Hence, among the twist terms appearing in \(C\) and \(D\), exactly
\(k-s\) twist terms from \(C\) and \(s\) twist terms from \(D\) have
exponents in
\(
\{n-k,\ldots,n-1\}.
\)
Since all twist exponents are pairwise distinct, these \(k\) exponents
occupy this \(k\)-element set exactly.

Moreover, each of these high-degree Vandermonde basis vectors occurs
in exactly one twist term among the generator rows of \(C\) and \(D\).
More precisely, for \(s+1\le i\le k\),
\(\mathbf{V}_{k-1+t_i}\) occurs only in the twist term of the
\(i\)-th row of \(C\), while for \(1\le j\le s\),
\(\mathbf{V}_{n-k-1+\gamma_j}\) occurs only in the twist term of the
\(j\)-th row of \(D\). Therefore, comparing the corresponding
coefficients in~(\ref{eq:ddw8}) gives
\[
\eta_i a_i=0\text{ and }\delta_j b_j=0,
\]
where \(s+1\le i\le k\),\text{ and }\(1\le j\le s\).
Since every \(\eta_i\) and \(\delta_j\) is nonzero, it follows that
\(a_{s+1}=\cdots=a_k=0\) and \(b_1=\cdots=b_s=0\).

For \(1\le i\le s\), neither the twist terms of \(C\) nor those of
\(D\) contribute to the coefficient of \(\mathbf{V}_{i-1}\).
Hence, comparison of the coefficient of \(\mathbf{V}_{i-1}\) gives
\(
a_i+b_i=0.
\)
Since \(b_i=0\), we obtain
\(
a_1=\cdots=a_s=0.
\)

Similarly, for \(s<i\le k\), comparison of the coefficient of
\(\mathbf{V}_{i-1}\) gives
\(
a_i+b_i=0.
\)
Since \(a_i=0\) for \(s<i\le k\), we obtain
\(
b_{s+1}=\cdots=b_k=0.
\)
Thus,
\(a_1=\cdots=a_k=0\) and \(b_1=\cdots=b_k=0\).
Substituting these equalities into~(\ref{eq:ddw8}), the linear relation
reduces to
\[
\sum_{j=k+1}^{n-k} b_j\mathbf{V}_{j-1}=0.
\]
Since
\(
\mathbf{V}_k,\mathbf{V}_{k+1},\ldots,
\mathbf{V}_{n-k-1}
\)
are linearly independent, it follows that
\(
b_{k+1}=\cdots=b_{n-k}=0.
\)
Consequently,~(\ref{eq:ddw8}) has only the zero solution.

\medskip
\noindent\textbf{Case (iii).}
For \(1\le i\le j\), the exponent of the twist term in \(C\) satisfies
\(
k-1+t_i\le n-k-1.
\)
Since \(\gamma_i=i\), the twist term in the \(i\)-th row of \(D\) has exponent
\(
n-k-1+i.
\)

For \(i\le j\), none of the twist terms in \(C\) has this exponent. Hence, comparison at \(\mathbf{V}_{n-k-1+i}\) gives
\(
b_i=0.
\)
Then the coefficient of \(\mathbf{V}_{i-1}\) gives
\(
a_i+b_i=0,
\)
and therefore \(a_i=0\).

For \(j+1\le i\le r\), the assumption \(t_i=n-2k+i\) gives
\(
k-1+t_i=n-k-1+i.
\)
Thus, the \(i\)-th twist terms in \(C\) and \(D\) have the same exponent. The corresponding coefficient equations are
\[
a_i+b_i=0 \text{ and } \eta_i a_i+\delta_i b_i=0.
\]
Equivalently,
\[
\begin{pmatrix}
1&1\\
\eta_i&\delta_i
\end{pmatrix}
\begin{pmatrix}
a_i\\
b_i
\end{pmatrix}
=0.
\]
Its determinant is \(\delta_i-\eta_i\ne0\), so \(a_i=b_i=0\).

For \(r<i\le k\), the row
\(
\mathbf{C}_i=\mathbf{V}_{i-1}
\)
contains no twist term, whereas the \(i\)-th twist term in \(D\) has exponent \(n-k-1+i\), which does not occur among the twist exponents of \(C\). Thus, comparison at \(\mathbf{V}_{n-k-1+i}\) gives \(b_i=0\), and comparison at \(\mathbf{V}_{i-1}\) gives \(a_i=0\).

Finally, the remaining untwisted rows of \(D\) correspond to the remaining Vandermonde basis vectors, so their coefficients also vanish. Therefore, (\ref{eq:ddw8}) has only the zero solution.

\medskip
\noindent\textbf{Case (iv).}
For every \(1\le i\le r\),
\(
k-1+t_i=n-k-1+i.
\)
Hence the twist terms in \(C\) and \(D\) contribute to the same Vandermonde basis vector \(\mathbf{V}_{n-k-1+i}\). The corresponding coefficient equations are identical to the \(2\times2\) system in Case (iii). Since \(\eta_i\ne\delta_i\), we obtain
\(
a_i=b_i=0.
\)
For \(r<i\le k\), the corresponding row of \(D\) contains a twist term with exponent \(n-k-1+i\), whereas no twist term in \(C\) has this exponent. Thus, comparison at \(\mathbf{V}_{n-k-1+i}\) gives \(b_i=0\), and comparison at \(\mathbf{V}_{i-1}\) gives \(a_i=0\).

Finally, all remaining coefficients vanish. Hence, (\ref{eq:ddw8}) has only the zero solution.

\medskip
\noindent\textbf{Case (v).}
Since
\(
1\le t_1<\cdots<t_k
\),\text{ and }
\(
n\le 2k-1+t_1,
\)
all \(k\) distinct exponents
\(
k-1+t_1,\ldots,k-1+t_k
\)
belong to the \(k\)-element set
\(
\{n-k,\ldots,n-1\}.
\)
Consequently,
\(
k-1+t_i=n-k+i-1,
\)
or equivalently,
\(
t_i=n-2k+i.
\)

For \(1\le i\le s\), the condition \(\gamma_i=i\) implies that the twist terms in \(C\) and \(D\) correspond to the same Vandermonde basis vector \(\mathbf{V}_{n-k+i-1}\). The corresponding coefficient equations form the same \(2\times2\) system as in Case (iii). Since
\(
\delta_i-\eta_i\ne0,
\) we obtain
\(
a_i=b_i=0.
\)

For \(s<i\le k\), the exponent \(n-k+i-1\) appears among the twist exponents of \(C\), but not among those of \(D\). Therefore, comparison at \(\mathbf{V}_{n-k+i-1}\) gives \(a_i=0\), and comparison at \(\mathbf{V}_{i-1}\) gives \(b_i=0\).

Finally, the remaining untwisted rows of \(D\) correspond to the remaining Vandermonde basis vectors, and their coefficients also vanish. Hence, (\ref{eq:ddw8}) has only the zero solution.

In every case, the stacked generator matrix is nonsingular, and therefore \((C,D)\) is an LCP.
\end{proof}

\begin{remark}
The proof of Theorem~\ref{thm:01} relies on the uniqueness of the
Vandermonde representation. In each case, the essential point is that
every high-degree basis vector
\(
\mathbf{V}_{n-k},\ldots,\mathbf{V}_{n-1}
\)
is either occupied by exactly one twist term or by a pair of twist terms
whose coefficient matrix is nonsingular. This guarantees that all
coefficients in a linear relation vanish.
\end{remark}

We now present a construction of MDS LCPs.

\begin{lemma}[{\cite[Theorem~1]{BeelenEtAl2018}}]\label{lem:kno}
Let \(1\le \ell\le k\), and let
\[
\F_{s_0}\subsetneq\F_{s_1}\subsetneq\cdots\subsetneq\F_{s_\ell}=\F_q
\]
be a chain of finite fields. Suppose \(k<n\le s_0\), and let
\(\alpha_1,\ldots,\alpha_n\in\F_{s_0}\) be pairwise distinct.
Let \(\boldsymbol{t},\mathbf{h},\boldsymbol{\eta}\) be chosen as in
Definition~2.1 such that
\(
\eta_i\in\F_{s_i}\setminus\F_{s_{i-1}},
\text{ for }
1\le i\le\ell.
\)
Then
\(
C_{n,k}(\boldsymbol{\alpha},\boldsymbol{1},\boldsymbol{t},
\mathbf{h},\boldsymbol{\eta})
\)
is MDS.
\end{lemma}

\begin{remark}\label{rem:zero-twist}
When \(\ell=0\), the code
\(
C_{n,k}(\boldsymbol{\alpha},\boldsymbol{1},\boldsymbol{t},
\mathbf{h},\boldsymbol{\eta})
\)
reduces to a Reed-Solomon code and is MDS.
\end{remark}

\begin{corollary}\label{cor:MDSLCP}
Let \(C\) and \(D\) be the TRS codes defined in~(3), with
dimensions \(k\) and \(n-k\), respectively. Suppose that each of
\(C\) and \(D\) either has no twist or satisfies the hypotheses of
Lemma~\ref{lem:kno}. If \(C\) and \(D\) satisfy one of the
conditions in Theorem~3.4, then \((C,D)\) is an MDS LCP.
Moreover, its security parameter is
\[
\min\{d(C),d(D^\perp)\}=n-k+1,
\]
which is optimal.
\end{corollary}

\begin{proof}
By Lemma~\ref{lem:kno} and Remark~\ref{rem:zero-twist},
both \(C\) and \(D\) are MDS. Moreover, Theorem~3.4 shows that
\((C,D)\) is an LCP. Hence \((C,D)\) is an MDS LCP.

Since \(C\) is an \([n,k]\) MDS code,
\(
d(C)=n-k+1.
\)
Since \(D\) is an \([n,n-k]\) MDS code, its dual \(D^\perp\) is an
\([n,k]\) MDS code. Therefore,
\(
d(D^\perp)=n-k+1.
\)
Consequently,
\[
\min\{d(C),d(D^\perp)\}=n-k+1.
\]
By the Singleton bound, this is the largest possible value of the
security parameter. Hence the security parameter is optimal.
\end{proof}

\section{LCD twisted Reed--Solomon codes}

\renewcommand{\thetheorem}{\arabic{section}.\arabic{theorem}}
\setcounter{theorem}{0}

We now specialize the previous method to the case $D=C^\perp$. Consider
\[
C=C_{n,k}\bigl(
\boldsymbol{\alpha},\mathbf{1},
\boldsymbol{t},\boldsymbol{h},\boldsymbol{\eta}
\bigr),
\]
where $\boldsymbol{\eta}=(\eta_1,\ldots,\eta_k)$,
$
\boldsymbol{h}=(0,1,\ldots,k-1)
$ and
$
\boldsymbol{t}=(t_1,\ldots,t_k)
$
for
$
1\le t_1<\cdots<t_k\le n-k.
$

Assume that $\{\alpha_1,\ldots,\alpha_n\}$ is a multiplicative subgroup of $\F_q^*$. Thus,
$
\alpha_i^n=1
$
for all $i$. Since $n\mid(q-1)$, the characteristic of $\F_q$ does not divide $n$, so $1/n$ is well defined in $\F_q$.
The generator matrix of $C$ is
\begin{equation*}
G_C=
\begin{pmatrix}
\mathbf{V}_0+\eta_1\mathbf{V}_{k-1+t_1}\\
\mathbf{V}_1+\eta_2\mathbf{V}_{k-1+t_2}\\
\vdots\\
\mathbf{V}_{k-1}+\eta_k\mathbf{V}_{k-1+t_k}
\end{pmatrix},
\end{equation*}
that is, the $i$-th row of the generator matrix of $C$
can be written as
\begin{equation}\label{eq:CiLCD}
\mathbf{C}_i
=
\mathbf{V}_{i-1}+\eta_i\mathbf{V}_{k-1+t_i},
\end{equation}
where \(1\le i\le k\). We first recall the following duality result.

\begin{lemma}[{\cite[Theorem~2]{BeelenEtAl2018}}]\label{lem:dual}
Let the evaluation points form a multiplicative subgroup of
\(\mathbb{F}_q^*\). Define \(u_j=\frac{\alpha_j}{n}\) for \(1\le j\le n\),
and let
\(\mathbf u=(u_1,\ldots,u_n)\).
Then
\[
C_{n,k}(\boldsymbol{\alpha},\mathbf 1,
\boldsymbol t,\boldsymbol h,\boldsymbol\eta)^\perp
=
C_{n,n-k}(\boldsymbol{\alpha},\mathbf u,
k-\boldsymbol h,n-k-\boldsymbol t,
-\boldsymbol\eta).
\]
Here
\(
k-\boldsymbol h=(k-h_1,\ldots,k-h_k)
\)
\text{ and }
\(
n-k-\boldsymbol t=(n-k-t_1,\ldots,n-k-t_k)
\)
are, respectively, the twist vector and the hook vector of the dual
code.
\end{lemma}

We now apply Lemma~\ref{lem:dual} to the present setting, where
\(
\mathbf h=(0,1,\ldots,k-1).
\)
For \(1\le i\le k\), the corresponding twist and hook parameters of the
dual code are
\(
k-h_i=k-i+1
\)\text{ and }
\(
n-k-t_i,
\)
respectively. Hence, before applying the column multipliers, the
generator row of the dual code corresponding to the hook
\(n-k-t_i\) has the form
$
\mathbf{V}_{n-k-t_i}-\eta_i\mathbf{V}_{n-i}
$
since
\(
(n-k)-1+(k-i+1)=n-i.
\)

Because
\(\{\alpha_1,\ldots,\alpha_n\}\) is a multiplicative subgroup of
\(\F_q^*\), we have
\(
\alpha_j^n=1,
\)
for \(1\le j\le n\).
Moreover, the evaluation points are precisely the roots of
\(x^n-1\). Then
\[
\prod_{\ell\ne j}(\alpha_j-\alpha_\ell)
=
n\alpha_j^{\,n-1}
=
n\alpha_j^{-1},
\]
and hence the corresponding dual column multiplier is
\[
u_j
=
\left(
\prod_{\ell\ne j}(\alpha_j-\alpha_\ell)
\right)^{-1}
=
\frac{\alpha_j}{n}.
\]
Therefore, multiplication by the dual column multiplier sends
\(
\mathbf{V}_e \text{ to } \frac1n\mathbf{V}_{e+1}.
\)

Consequently, a parity-check matrix \(H_C\) of \(C\) may be chosen
with rows \(H_m\), where \(1\le m\le n-k\) and
\begin{equation}\label{eq:Hm}
\mathbf{H}_m=
\begin{cases}
\dfrac{1}{n}\mathbf{V}_m,
& m\ne n+1-k-t_i \text{ for all } 1\le i\le k,\\[1.2ex]
\dfrac{1}{n}\left(\mathbf{V}_m-\eta_i\mathbf{V}_{n-i+1}\right),
& m=n+1-k-t_i \text{ for some } 1\le i\le k.
\end{cases}
\end{equation}
Here \(\mathbf{V}_n=\mathbf{V}_0\) since \(\alpha_j^n=1\) for every \(j\).

By Lemma~2.4, the pair \((C,D)\) forms an LCP if and only if the stacked
generator matrix
\[
\begin{pmatrix}
G_C\\
G_D
\end{pmatrix}
\]
is nonsingular. In particular, when \(D=C^\perp\), the code \(C\) is
LCD if and only if
\[
\begin{pmatrix}
G_C\\
H_C
\end{pmatrix}
\]
is nonsingular. Therefore, as in Section~III, it suffices to study the
coefficient matrix of these rows with respect to the Vandermonde basis
\(
\mathbf{V}_0,\mathbf{V}_1,\ldots,\mathbf{V}_{n-1}.
\)

For convenience, define
\(
e_i=k-1+t_i, \text{ and }
m_i=n+1-k-t_i.
\)

\begin{theorem}\label{thm:LCD}
The code $C$ is LCD under any of the following conditions.
\begin{enumerate}
\item[(i)] $k\ge2$, $n\ge2k+t_k-1$, and
$t_i+t_{k+1-i}=n-2k+2$ for $1\le i\le k$.

\item[(ii)] $n=2k$, $1+\eta_1^2\ne0$, and
$\eta_i+\eta_{k+2-i}\ne0$ for $2\le i\le k$.

\item[(iii)] $2k<n\le2k-1+t_1$, and
$\eta_i+\eta_{k+2-i}\ne0$ for $2\le i\le k$.
\end{enumerate}
\end{theorem}

\begin{proof}
Consider a linear relation
\begin{equation}\label{eq:LCDrelation}
\sum_{i=1}^{k} a_i \mathbf{C}_i
+
\sum_{m=1}^{n-k} b_m \mathbf{H}_m
=0,
\end{equation}
where \(\mathbf{C}_i\) and \(\mathbf{H}_m\) are given in
(\ref{eq:CiLCD}) and (\ref{eq:Hm}), respectively.

Since \(\mathbf{V}_0,\ldots,\mathbf{V}_{n-1}\) form a basis of
\(\F_q^n\), the coefficient of each Vandermonde basis vector must
vanish separately.

Since the parity-check rows in (\ref{eq:Hm}) contain a common
nonzero factor \(1/n\), we multiply each parity-check row by \(n\).
This row scaling does not affect linear independence. For simplicity,
we keep the notation \(\mathbf{H}_m\) and \(b_m\) for the rescaled
parity-check rows and the corresponding coefficients, respectively.

\medskip
\noindent\textbf{Case (i).}
The inequality \(n\ge 2k+t_k-1\) is equivalent to
\(e_k\le n-k\). Since
\(
e_i=k-1+t_i
\)
and \(1\le t_1<\cdots<t_k\), every twist exponent appearing in the
generator matrix of \(C\) lies in
\(
\{k,\ldots,n-k\}.
\)

For \(2\le i\le k\), the twisted component of
\(\mathbf{H}_{m_i}\) occurs at
\(
\mathbf{V}_{n-i+1},
\)
with \(n-i+1\in\{n-k+1,\ldots,n-1\}\).
These basis vectors neither occur as twist components in the generator
matrix of \(C\) nor as ordinary components of the parity-check rows.
Moreover, the exponents \(n-i+1\), \(2\le i\le k\), are pairwise
distinct. Therefore, comparing the coefficient of
\(\mathbf{V}_{n-i+1}\) in~(\ref{eq:LCDrelation}) gives
\(
-\eta_i b_{m_i}=0,
\)
for \(2\le i\le k\).
Since \(\eta_i\ne0\), it follows that
\(
b_{m_i}=0.
\)

By the symmetry condition
\(
t_i+t_{k+1-i}=n-2k+2,
\)
together with the definitions
\(
e_i=k-1+t_i\text{ and }
m_i=n+1-k-t_i,
\)
we obtain
\[
\begin{aligned}
m_{k+1-i}=n+1-k-t_{k+1-i}=k-1+t_i=e_i.
\end{aligned}
\]
Hence,
$
e_i=m_{k+1-i}
$
for \(1\le i\le k\).

For \(1\le i\le k-1\), comparison of the coefficient of
\(\mathbf{V}_{e_i}\) in~(\ref{eq:LCDrelation}) gives
\(
\eta_i a_i+b_{m_{k+1-i}}=0.
\)
Since \(k+1-i\ge2\), we have already shown that
\(
b_{m_{k+1-i}}=0.
\)
Thus,
\(a_i=0\) for \(1\le i\le k-1\).
In particular, \(a_1=0\) since \(k\ge2\).

Next, since the twisted component of \(\mathbf{H}_{m_1}\) occurs at
\(\mathbf{V}_n=\mathbf{V}_0\), comparison of the coefficient of
\(\mathbf{V}_0\) gives
\(
a_1-\eta_1b_{m_1}=0.
\)
Hence,
\(
b_{m_1}=0.
\)

Finally, since \(e_k=m_1\), comparison of the coefficient of
\(\mathbf{V}_{e_k}\) gives
\(
\eta_k a_k+b_{m_1}=0,
\)
and therefore
\(
a_k=0.
\)
Thus,
\(
a_1=\cdots=a_k=0
\)\text{ and }\(
b_{m_i}=0
\)
for \(1\le i\le k\).

Substituting these equalities into~(\ref{eq:LCDrelation}), the remaining
relation is
\[
\sum_{\substack{1\le m\le n-k\\
m\notin\{m_1,\ldots,m_k\}}}
b_m\mathbf{V}_m=0.
\]
Since the corresponding Vandermonde basis vectors are linearly
independent, we obtain $b_m=0$ for all \(m\notin\{m_1,\ldots,m_k\}\).

Therefore, all coefficients in~(\ref{eq:LCDrelation}) are zero.
Hence,
\[
\begin{pmatrix}
G_C\\
H_C
\end{pmatrix}
\]
is nonsingular, and thus \(C\) is LCD.

\medskip
\noindent\textbf{Case (ii).}
Since $n=2k$ and $t_1\ge1$, we have $e_1\ge k=n-k$. The $k$ distinct exponents
\(
e_1<\cdots<e_k
\)
all belong to the $k$-element set $\{k,k+1,\ldots,2k-1\}$, so they must be exactly this set. Hence,
\[
e_i = k+i-1,\; t_i = i\text{ and } m_i = k+1-i.
\]
For $i=1$, we have
\[
\mathbf{C}_1=\mathbf{V}_0+\eta_1\mathbf{V}_k\text{ and }
n\mathbf{H}_k=\mathbf{V}_k-\eta_1\mathbf{V}_0.
\]
The corresponding coefficient block, with respect to the coordinates $(\mathbf{V}_0,\mathbf{V}_k)$, is
\[
\begin{pmatrix}
1&-\eta_1\\
\eta_1&1
\end{pmatrix},
\]
whose determinant is
$
1+\eta_1^2\ne0.
$
For $2\le i\le k$, put $p=k+2-i$. Then $m_p=i-1$ and $n-p+1=e_i$. Therefore
\[
\mathbf{C}_i=\mathbf{V}_{i-1}+\eta_i\mathbf{V}_{e_i},
\]
while
\[
n\mathbf{H}_{i-1}=\mathbf{V}_{i-1}-\eta_{k+2-i}\mathbf{V}_{e_i}.
\]
The corresponding coefficient block is
\[
\begin{pmatrix}
1&1\\
\eta_i&-\eta_{k+2-i}
\end{pmatrix},
\]
whose determinant equals
\(
-(\eta_i+\eta_{k+2-i})\ne0.
\)
Thus every block is nonsingular. Consequently, all coefficients in the
linear relation (\ref{eq:LCDrelation}) are zero. Hence,
\[
\begin{pmatrix}
G_C\\
H_C
\end{pmatrix}
\]
is nonsingular, and therefore \(C\) is LCD.

\medskip
\noindent\textbf{Case (iii).}
The hypothesis $n\le2k-1+t_1$ is equivalent to $e_1\ge n-k$. Since the $e_i$ are distinct and all lie in the $k$-element set $\{n-k,\ldots,n-1\}$, they fill this entire set. Therefore,
\[
e_i=n-k+i-1,
\;
t_i=n-2k+i \text{ and }
m_i=k+1-i.
\]
For $2\le i\le k$, exactly the same $2\times2$ blocks as in Case (ii) occur, and their determinants are
\(
-(\eta_i+\eta_{k+2-i})\ne0.
\)
It remains to handle $i=1$. Here,
\[
\mathbf{C}_1=\mathbf{V}_0+\eta_1\mathbf{V}_{n-k} \text{ and }n\mathbf{H}_k=\mathbf{V}_k-\eta_1\mathbf{V}_0.
\]
Because \(n>2k\), we have \(k<n-k\). Hence, the coordinate
\(\mathbf{V}_k\) appears only in \(\mathbf{H}_k\) among the exceptional rows and
cannot occur as the exponent of a twist term in the generator matrix of \(C\). Comparing the coefficient of \(\mathbf{V}_k\) in (\ref{eq:LCDrelation}) therefore gives $b_k=0$. The coefficient of \(\mathbf{V}_0\) then gives $a_1-\eta_1b_k=0$,
and hence \(a_1=0\). For \(2\leq i\leq k\), the nonsingularity of the corresponding \(2\times2\) blocks gives
\[
a_i=0\ \text{and}\ b_{i-1}=0.
\]
Consequently,
\[
a_1=\cdots=a_k=0\ \text{and}\ b_1=\cdots=b_k=0.
\]
Moreover, since $m_i=k+1-i,\ 1\leq i\leq k$, we have $\{m_1,\ldots,m_k\}=\{1,\ldots,k\}$. Therefore, for every \(k+1\leq m\leq n-k\), the row \(\mathbf{H}_m\)
contains no twist term and satisfies
\[
n\mathbf{H}_m=\mathbf{V}_m.
\]
After substituting the preceding equalities into (\ref{eq:LCDrelation}), the
remaining linear relation becomes
\[
\sum_{m=k+1}^{n-k} b_m\mathbf{V}_m=0.
\]
Since $\mathbf{V}_{k+1},\mathbf{V}_{k+2},\ldots,\mathbf{V}_{n-k}$ are linearly independent, it follows that
\[
b_{k+1}=\cdots=b_{n-k}=0.
\]
Thus all the coefficients in the linear relation (\ref{eq:LCDrelation}) are
zero. Hence,
\[
\begin{pmatrix}
G_C\\
H_C
\end{pmatrix}
\]
is nonsingular, and therefore \(C\) is LCD.

In all three cases $C\cap C^\perp=\{\mathbf{0}\}$.
\end{proof}

\begin{remark}
The restriction $k\ge2$ in Theorem~4.2(i) is essential for the proof above. If $k=1$, the symmetry condition alone does not force the code to be LCD; an additional condition
\[
1+\eta_1^2\ne0
\]
is needed in the corresponding two-coordinate block.
\end{remark}

\begin{corollary}
Let $C$ satisfy one of the conditions of Theorem~\ref{thm:LCD}. Suppose, in addition, that $C$ satisfies all the hypotheses of Lemma~\ref{lem:kno}. Then $C$ is an MDS LCD twisted Reed--Solomon code. Indeed, Lemma~\ref{lem:kno} implies that $C$ is MDS, while Theorem~\ref{thm:LCD} implies that
\[
C\cap C^\perp=\{\mathbf{0}\}.
\]
Therefore,
\(
d(C)=n-k+1
\)
and the corresponding LCD construction attains the largest possible minimum distance for its length and dimension.
\end{corollary}

\section{Concluding remarks}

In this paper, we study linear complementary pairs (LCPs) constructed from twisted Reed--Solomon (TRS) codes. By representing the generator matrices in the Vandermonde basis, we transform the LCP condition into the nonsingularity problem of a sparse coefficient matrix. Based on this approach, we derive necessary conditions and several sufficient criteria for two TRS codes to form an LCP, according to different overlap patterns of the twist exponents. We further investigate the LCD property of TRS codes and construct several families of LCD TRS codes. Combining these results with known MDS conditions, we obtain MDS LCD codes and MDS LCPs with optimal security parameters. Our results provide a systematic method for studying complementary pairs of twisted Reed--Solomon codes and may lead to further constructions with more general twist configurations.

\end{document}